\documentclass[11pt]{article}

\usepackage[T1]{fontenc}
\usepackage[sc]{mathpazo}
\usepackage{amsmath}
\usepackage{amssymb}
\usepackage{enumerate}
\usepackage{amsthm}
\usepackage{amsfonts,mathrsfs}
\usepackage[style]{fncychap}
\usepackage{graphicx} 
\usepackage{geometry} 
\usepackage{eepic}
\usepackage{ifthen}
\newboolean{ElectronicVersion}
\setboolean{ElectronicVersion}{true} 

\usepackage{hyperref}
\hypersetup{pdfpagemode=UseNone}

\newcommand{\setft}[1]{\mathrm{#1}}
\newcommand{\lin}[1]{\setft{L}\left(#1\right)}

\def\I{\mathbb{1}}

\newenvironment{mylist}[1]{\begin{list}{}{
    \setlength{\leftmargin}{#1}
    \setlength{\rightmargin}{0mm}
    \setlength{\labelsep}{2mm}
    \setlength{\labelwidth}{8mm}
    \setlength{\itemsep}{0mm}}}
    {\end{list}}

\newcommand{\defeq}{\stackrel{\smash{\textnormal{\tiny def}}}{=}}

\def\cX{\mathcal{X}}\def\cY{\mathcal{Y}}

\def\rR{\mathrm{R}}

\def\A{\textsf{A}}

\newtheorem{thrm}{Theorem}[section]
\newtheorem{lem}[thrm]{Lemma}

\theoremstyle{definition}

\newtheorem{remark}[thrm]{Remark}

\numberwithin{equation}{section}

\newcounter{questionnumber}

\begin{document}

\title{An optimal problem for relative entropy}

\author{Fan Wang, Jun Zhu\footnote{E-mail: junzhu@yahoo.cn} and Lin Zhang\\[1mm]
  {\it\small Institute of Mathematics, Hangzhou Dianzi University, Hangzhou 310018,
  PR~China}}

\date{}
\maketitle

\begin{abstract}

Relative entropy is an essential tool in quantum information theory. There are so many problems which are related to relative entropy. In this article, the optimal values which are defined by $\displaystyle\max_{U\in{U(\cX_{d})}} S(U\rho{U^{\ast}}\parallel\sigma)$ and $\displaystyle\min_{U\in{U(\cX_{d})}} S(U\rho{U^{\ast}}\parallel\sigma)$ for two positive definite operators $\rho,\sigma\in{\textmd{Pd}(\cX)}$ are obtained.  And the set of $S(U\rho{U^{\ast}}\parallel\sigma)$ for every unitary operator $U$ is full of the interval $[\displaystyle\min_{U\in{U(\cX_{d})}} S(U\rho{U^{\ast}}\parallel\sigma),\displaystyle\max_{U\in{U(\cX_{d})}} S(U\rho{U^{\ast}}\parallel\sigma)]$.

\end{abstract}

\section{Introduction}

In complex Euclidean spaces $\cX$, one writes $\lin{\cX,\cY}$ to refer to the collection of all linear mappings of the form
 $\A:\cX\rightarrow\cY$, denote $\lin{\cX,\cX}$ by $\lin{\cX}$, an operator $A\in\lin{\cX}$ is normal if and only if it commutes with its adjoint: $[A,A^\ast]$ , or equivalently $AA{^\ast}=A^{\ast}A$. An operator is \emph{positive semidefinite} $A\in\lin{\cX}$ if and only if it holds that $A=B^{\ast}B$ for some operator $B\in\lin{\cX}$. The collection of operators is denoted $\textmd{Pos}(\cX)=\{B^{\ast}B:B\in\lin{\cX}\}$. The notation $P\geq0$ is also used to mean that $P$ is positive semidefinite , while $A\geq{B}$ means that $A-B$ is positive semidefinite. A positive semidefinite operator $P\in{\textmd{Pos}(\cX)}$ is said to be positive definite if, in addition to being positive semidefinite, it is non-singular. We write
 $\textmd{Pd}(\cX)=\{A\in{\textmd{Pos}(\cX)}:det(A)\neq0\}$ to denote the set of such operators for a given complex Euclidean space $\cX$.

An unitary operator $A\in\lin{\cX,\cY}$ is a linear isometry if it preserves the Euclidean norm-meaning that $\parallel{Au}\parallel=\parallel{u}\parallel$ for all $u\in{\cX}$. The condition that $\parallel{Au}\parallel=\parallel{u}\parallel$ for all $u\in{\cX}$ is equivalent to $AA{^\ast}=\I_{\cX}$. The collection of unitary operator is denoted $U(\cX,\cY)=\{A\in\lin{\cX,\cY}:AA{^\ast}=\I_{\cX}\}$.

\begin{thrm}[Spectral Theorem\label{Lind}(\cite{{Watrous}})]
Let $\cX$ be a complex Euclidean space, let $A\in\lin{\cX}$ be a normal operator, and assume that the distinct eigencalues of $A$ are $\lambda_{1},\ldots,\lambda_{k}$. Then  there exist an orthogonal basis $\{x_{1},x_{2},\ldots,x_{n}\}$ of $\cX$, such that
$$
A=\sum_{i=1}^k \lambda_{i}x_{i}x_{i}^\ast.
$$
\end{thrm}

von Neumann entropy and relative entropy are powerful tools in quantum information theory. The quantum relative entropy(or relative von Neumann entropy) is indispensable as a tool for the von Neumann entropy.

For two positive definite operators $P,Q\in{\textmd{Pd}(\cX)}$, then the von Neumann entropy of $P$ is defined by
$$
S(P)\defeq \textmd{Tr}(P\log(P)),
$$
the quantum relative entropy between $P$ and $Q$ is defined by:
$$
S(P\parallel{Q})=\textmd{Tr}(P\log(P))-\textmd{Tr}(P\log(Q)),
$$
we take the base 2 logarithm of corresponding eigenvalues. Specifically,when $im(P)\subsetneq{kerQ}$, we define $S(P\parallel{Q})=\infty$.

Let $B=[b_{ij}]$ be a $N\times{N}$ \emph{bi-stochastic matrix}, that is, $b_{ij}\geq0$, and
$$
\sum_{i=1}^k b_{ij}=\sum_{j=1}^k b_{ij}=1,
$$
for each $i,j=1,\ldots,N$. We only consider finite-dimensional complex Hilbert spaces, let $\mathbf{r},\mathbf{s}\in\rR^n$, $\mathbf{r}$ is \emph{majorised} by $\mathbf{s}$, written as $\mathbf{r}\prec\mathbf{s}$, if
$$
\sum_{i=1}^k r_{i}\leq{\sum_{j=1}^k s_{i}},\quad \quad 1\leq{k}\leq{n-1}
$$
and
$$
\sum_{i=1}^n r_{i}=\sum_{j=1}^n s_{i}.
$$

\section{Description of the problems and main results}
For two positive definite operators $P,Q\in{\textmd{Pd}(\cX)}$, here $\textmd{dim}(\cX_{d})=d$. We try to find the optimal values which are defined by
$$
\displaystyle\max_{U\in{U(\cX_{d})}} S(U\rho{U^{\ast}}\parallel\sigma) \quad and \quad \displaystyle\min_{U\in{U(\cX_{d})}} S(U\rho{U^{\ast}}\parallel\sigma)
$$
respectively. Since
$$
\displaystyle\max_{U\in{U(\cX_{d})}} S(U\rho{U^{\ast}}\parallel\sigma)=-S(\rho)-\displaystyle\min_{U\in{U(\cX_{d})}} \textmd{Tr}(U\rho{U^{\ast}}\log\sigma)
$$
and
$$
\displaystyle\min_{U\in{U(\cX_{d})}} S(U\rho{U^{\ast}}\parallel\sigma)=-S(\rho)-\displaystyle\max_{U\in{U(\cX_{d})}} \textmd{Tr}(U\rho{U^{\ast}}\log\sigma)
$$
Thus the above optimal problems are reduced to the followings:
$$
\displaystyle\max_{U\in{U(\cX_{d})}} \textmd{Tr}(U\rho{U^{\ast}}\log\sigma) \quad and \quad \displaystyle\min_{U\in{U(\cX_{d})}} \textmd{Tr}(U\rho{U^{\ast}}\log\sigma)
$$

Firstly, we assume that $\mathbf{supp}(\rho)\subseteq\mathbf{supp}(\sigma)$ and $\sigma$ is of full rank. By the spectral decomposition, we have that
$$
\sigma=\displaystyle\sum_{j} \lambda_{j}^\downarrow(\sigma)|y_{j}\rangle\langle{y_{j}}|=\displaystyle\sum_{j} \lambda_{j}^\uparrow(\sigma)|z_{j}\rangle\langle{z_{j}}|.
$$
For an orthonormal basis $\{|y_{j}\rangle\}$ of $\cX_{d}$ and $|y_{k}\rangle=|z_{d+1-k}\rangle$ for $k\in[d]=\{1,\ldots,d\}$.
Here $\lambda_{j}^\downarrow(\sigma)$ stands for eigenvalues are arranged in decreasing order and $\lambda_{j}^\uparrow(\sigma)$ stands for eigenvalues are arranged in increasing order.
If $[\rho,\sigma]\neq0$, from the spectral decomposition of $\rho$:
$$
\rho=\displaystyle\sum_{j} \lambda_{j}^\downarrow(\rho)|x_{j}\rangle\langle{x_{j}}|,
$$
where $\{|x_{j}\rangle\}$ is an orthonormal basis of $\cX_{d}$, we know that there exist two unitary operators $U,V:U=\displaystyle\sum_{j}|y_{j}\rangle\langle{x_{j}}|$ and $V=\displaystyle\sum_{j}|z_{j}\rangle\langle{x_{j}}|$, such that $U|x_{j}\rangle=|y_{j}\rangle$ for each $j$ , respectively.

Now
$$
\textmd{Tr}(U\rho{U^{\ast}}\log\sigma)=\displaystyle\sum_{j} \lambda_{j}^\downarrow(\rho)\log\lambda_{j}^\downarrow(\sigma)
$$
and
$$
\textmd{Tr}(V\rho{V^{\ast}}\log\sigma)=\displaystyle\sum_{j} \lambda_{j}^\downarrow(\rho)\log\lambda_{j}^\uparrow(\sigma).
$$

In mathematics, the rearrangement inequality states that
\begin{equation}
\begin{split}
x_{n}y_{1}+\ldots+x_{1}y_{n}\leq{x_{\pi(1)}y_{1}+\ldots+x_{\pi(n)}y_{n}}\leq{x_{1}y_{1}+\ldots+x_{n}y_{n}}
\end{split}
\end{equation}
for every choice of real numbers
$$
x_{1}\leq\ldots\leq{x_{n}} \quad and \quad y_{1}\leq\ldots\leq{y_{n}},
$$
and every permutation $x_{\pi(1)},\ldots,x_{\pi(n)}$ of $x_{1}\leq\ldots\leq{x_{n}}$. If the number are different, meaning that
$$
x_{1}<\ldots<{x_{n}} \quad and \quad y_{1}<\ldots<{y_{n}},
$$
then the lower bound is attained only for the permutation which reverses the order, i.e. $\pi(i)=n-i+1$ for all $i=1,\ldots,n$.

Note that the rearrangement inequality makes no assumptions on the signs of the real number.
\begin{lem}
Let $\mathbf{x}^\downarrow,\mathbf{y}^\downarrow,\mathbf{u}^\downarrow\in\rR^d$, all having their coordinates in decreasing order. If  $\mathbf{y}^\downarrow\prec\mathbf{x}^\downarrow$, i.e. $\mathbf{y}^\downarrow$ is \emph{majorized} by $\mathbf{x}^\downarrow$, then
$$
\langle\mathbf{u}^\downarrow,\mathbf{y}^\downarrow\rangle\leq\langle\mathbf{u}^\downarrow,\mathbf{x}^\downarrow\rangle.
$$
\end{lem}
\begin{proof}
For the completeness, we give the detailed proof. We make an induction argument in the following.

We denote the $\textmd{Tr}(\mathbf{x}^\downarrow)=\textmd{Tr}(\mathbf{y}^\downarrow)=s$, i.e. ,
$s=\displaystyle\sum_{j=1}^{d}\mathbf{x}_{j}^\downarrow=\displaystyle\sum_{j=1}^{d}\mathbf{y}_{j}^\downarrow$.If  $d=2$, by a simple algebra, we have
$$
\langle\mathbf{u}^\downarrow,\mathbf{y}^\downarrow\rangle\leq\langle\mathbf{u}^\downarrow,\mathbf{x}^\downarrow\rangle.
$$

In fact,
$$
\sum_{j=1}^2\mathbf{x}_{j}^\downarrow\mathbf{u}_{j}^\downarrow-\sum_{j=1}^2\mathbf{x}_{j}^\downarrow\mathbf{u}_{j}^\downarrow=
(su_{1}^\downarrow+x_{2}^\downarrow(u_{2}^\downarrow-
(su_{1}^\downarrow))-(su_{1}^\downarrow+y_{2}^\downarrow(u_{2}^\downarrow-u_{1}^\downarrow))=
(x_{2}^\downarrow-y_{2}^\downarrow)(u_{2}^\downarrow-u_{1}^\downarrow)\geq0
$$
Where the last line in the above inequality holds because $\mathbf{y}^\downarrow$ is \emph{majorized} by $\mathbf{x}^\downarrow$, i.e. ,
$$
x_{1}^\downarrow\geq{y_{1}^\downarrow}, \quad
s=\displaystyle\sum_{j}\mathbf{x}^\downarrow=\displaystyle\sum_{j}\mathbf{y}^\downarrow.
$$

Next we suppose that the inequality holds for $d-1$. Now we show that the inequality holds for $d$.
$$
\sum_{j=1}^d\mathbf{x}_{j}^\downarrow\mathbf{u}_{j}^\downarrow-\sum_{j=1}^d\mathbf{x}_{j}^\downarrow\mathbf{u}_{j}^\downarrow=
\sum_{j=1}^{d-1}(x_{j}^\downarrow-y_{j}^\downarrow)(u_{j}^\downarrow-u_{d}^\downarrow)
\geq(u_{d-1}^\downarrow-u_{d}^\downarrow)\sum_{j=1}^{d-1}(x_{j}^\downarrow-y_{j}^\downarrow)\geq0
$$
The proof is completed.
\end{proof}
\begin{remark}
The result obtained in Lemma 2.1 can be reformulated as:For two fixed d-dimensional real vector
$\mathbf{x}^\downarrow,\mathbf{u}^\downarrow\in\rR^d$, we have
$$
\displaystyle\max_{B} \langle\mathbf{u}^\downarrow\mid{B}\mid\mathbf{x}^\downarrow\rangle
=\langle\mathbf{u}^\downarrow,\mathbf{x}^\downarrow\rangle,
$$
where the above optimization is taken over the set of all bi-stochastic matrices of size $d\times{d}$.
\end{remark}

By employing the rearrangement inequality, we have that, for any permutation $\pi\in{S_{d}}$,
$$
\displaystyle\sum_{j} \lambda_{j}^\downarrow(\rho)\log\lambda_{j}^\downarrow(\sigma)
\geq\displaystyle\sum_{j} \lambda_{j}^\downarrow(\rho)\log\lambda_{\pi(j)}^\downarrow(\sigma)
\geq\displaystyle\sum_{j} \lambda_{j}^\downarrow(\rho)\log\lambda_{j}^\uparrow(\sigma).
$$

Thus from the above discussion , it is seen that the above questions are reduced to the case where $[\rho,\sigma]=0$
 and $\sigma>0$.

Now we can make an assumption:
$$
\rho=\displaystyle\sum_{j} \lambda_{j}^\downarrow(\rho)|j\rangle\langle{j}| \quad and \quad  \sigma=\displaystyle\sum_{j} \lambda_{j}^\downarrow(\sigma)|j\rangle\langle{j}|,
$$
where $\lambda_{j}^\downarrow(\sigma)>0$ for each $j$. We still show that

\begin{equation}
\begin{split}
\displaystyle\max_{U\in{U(\cX_{d})}} \textmd{Tr}(U\rho{U^{\ast}}\log\sigma)=\displaystyle\sum_{j} \lambda_{j}^\downarrow(\rho)\log\lambda_{j}^\downarrow(\sigma)
\end{split}
\end{equation}
\begin{equation}
\begin{split}
\displaystyle\min_{U\in{U(\cX_{d})}} \textmd{Tr}(U\rho{U^{\ast}}\log\sigma)=\displaystyle\sum_{j} \lambda_{j}^\downarrow(\rho)\log\lambda_{j}^\uparrow(\sigma)
\end{split}
\end{equation}

Now for any unitary operator $W\in{U(\cX_{d})}$,
$$
\textmd{Tr}(W\rho{W^{\ast}}\log\sigma)=\displaystyle\sum_{i,j} \lambda_{i}^\downarrow(\rho)\log\lambda_{j}^\downarrow(\sigma)
{\mid\langle{j}|W|{i}\rangle\mid}^{2}.
$$

Setting $D=[{\mid\langle{j}|W|{i}\rangle\mid}^{2}]$, we see that $D$ is a bi-stochastic matrix since $D=W\circ\overline{W}$,
where $\circ$ is \emph{Schur product}.

Now we have $\textmd{Tr}(W\rho{W^{\ast}}\log\sigma)=\displaystyle\sum_{j} \lambda_{j}^\downarrow(\rho)\alpha_{j}$, where $\alpha_{j}=\displaystyle\sum_{i} \mid\langle{j}|W|{i}\rangle\mid\lambda_{i}^\downarrow(\rho)=\displaystyle\sum_{i} D_{i,j}\lambda_{i}^\downarrow(\rho)$. This indicates that
$$
\alpha=D\lambda^\downarrow(\rho),
$$
where
$$
\alpha\defeq [\alpha_{1},\ldots,\alpha_{d}]^\intercal  \quad and \quad \lambda^\downarrow(\rho)\defeq [\lambda_{1}^\downarrow(\rho),\ldots,\lambda_{d}^\downarrow(\rho)]^\intercal,
$$
 which implies that
$$
\alpha\prec\lambda^\downarrow(\rho).
$$
Therefore
$$
\textmd{Tr}(W\rho{W^{\ast}}\log\sigma)=\langle\log\lambda^\downarrow(\rho)|D|\lambda^\downarrow(\rho)\rangle,
$$
where $\log\lambda^\downarrow(\rho)\defeq [\log\lambda_{1}^\downarrow(\rho),\ldots,\log\lambda_{d}^\downarrow(\rho)]^\intercal$.

In what follows,we show that
$$
\displaystyle\max_{U\in{U(\cX_{d})}} \textmd{Tr}(W\rho{W^{\ast}}\log\sigma)=
\displaystyle\max_{D}\langle\log\lambda^\downarrow(\rho)|D|\lambda^\downarrow(\rho)\rangle.
$$
It follows from Lemma 2.1 that
$$
\displaystyle\max_{D}\langle\log\lambda^\downarrow(\rho)|D|\lambda^\downarrow(\rho)\rangle
=\langle\log\lambda^\downarrow(\rho),\lambda^\downarrow(\rho)\rangle=\displaystyle\sum_{j=1}^{d} \lambda_{j}^\downarrow(\rho)\log\lambda_{j}^\downarrow(\sigma).
$$
Thus we are arriving at the following conclusion:
\begin{thrm}\label{th:Sentropy:equality}
Let $\cX$ be a complex Euclidean spaces and assume that $\rho,\sigma\in{\textmd{Pd}(\cX)}$ are positive definite operators on $\cX$. Then
$$
\displaystyle\max_{U\in{U(\cX_{d})}} S(U\rho{U^{\ast}}\parallel\sigma)=-S(\rho)-\langle\log\lambda^\downarrow(\rho),\lambda^\downarrow(\rho)\rangle.
$$

\end{thrm}

Similarly,we have
\begin{thrm}\label{th:Sentropy:equality}
Let $\cX$ be a complex Euclidean spaces and assume that $\rho,\sigma\in{\textmd{Pd}(\cX)}$ are positive definite operators on $\cX$. Then
$$
\displaystyle\min_{U\in{U(\cX_{d})}} S(U\rho{U^{\ast}}\parallel\sigma)=-S(\rho)-\langle\log\lambda^\downarrow(\rho),\lambda^\uparrow(\rho)\rangle.
$$
\end{thrm}

\section{Further result}

In this section,  the value range of $S(U\rho{U^{\ast}}\parallel\sigma)$ is discussed.
We define the following sets
$$
M(\rho,\sigma)\defeq \{S(U\rho{U^{\ast}}\parallel\sigma):U\in{U(\cX_{d})}\},\quad \widetilde{M}(\rho,\sigma)\defeq
\{\textmd{Tr}(U\rho{U^{\ast}}\sigma):U\in{U(\cX_{d})}\}
$$

If the set $M(\rho,\sigma)$ is the interval
$$
[\displaystyle\min_{U\in{U(\cX_{d})}} S(U\rho{U^{\ast}}\parallel\sigma),\displaystyle\max_{U\in{U(\cX_{d})}} S(U\rho{U^{\ast}}\parallel\sigma)]
$$
it is equivalent to the following:
$$
\widetilde{M}=[\displaystyle\min_{U\in{U(\cX_{d})}} \textmd{Tr}(U\rho{U^{\ast}}\sigma),\displaystyle\max_{U\in{U(\cX_{d})}} \textmd{Tr}(U\rho{U^{\ast}}\sigma)]
$$

\begin{thrm}\label{th:Sentropy:equality}
Let $\cX$ be a complex Euclidean spaces and assume that $\rho,\sigma\in{\textmd{Pd}(\cX)}$ are positive definite operators on $\cX$. If we define the following sets
$$
\widetilde{M}(\rho,\sigma)\defeq\{\textmd{Tr}(U\rho{U^{\ast}}\sigma):U\in{U(\cX_{d})}\},
$$
then
$$
\widetilde{M}(\rho,\sigma)=[\displaystyle\min_{U\in{U(\cX_{d})}} \textmd{Tr}(U\rho{U^{\ast}}\sigma),\displaystyle\max_{U\in{U(\cX_{d})}} \textmd{Tr}(U\rho{U^{\ast}}\sigma)]
$$
\end{thrm}
\begin{proof}
This problem is equivalent to
$$
\widetilde{M}(\rho,\sigma)=[\displaystyle\min_{U\in{U(\cX_{d})}} \textmd{Tr}(U\rho{U^{\ast}}\sigma),\displaystyle\max_{U\in{U(\cX_{d})}} \textmd{Tr}(U\rho{U^{\ast}}\sigma)]=[\displaystyle\sum_{i=1}^{n} \lambda_{i}^\downarrow(\rho)\log\lambda_{i}^\uparrow(\sigma),\displaystyle\sum_{i=1}^{n} \lambda_{i}^\downarrow(\rho)\log\lambda_{i}^\downarrow(\sigma)].
$$
If $d=n$ we assume that $\lambda^\downarrow(\rho)=(\lambda_{1},\ldots,\lambda_{n}) \quad and  \quad \lambda^\downarrow(\sigma)=(\mu_{1},\ldots,\mu_{n})$.Then the problem is equivalent to
$$
\widetilde{M}(\rho,\sigma)=[\displaystyle\sum_{i=1}^{n} \lambda_{i}\mu_{n-i+1},\displaystyle\sum_{i=1}^{n}\lambda_{i}\mu_{i}].
$$

We make an induction argument in the following. And we only need consider the diagonal form of $\rho$ and $\sigma$ for any unitary, what's more,  the value on diagonal are arranged in decreasing order.

Firstly, if $d=2$, the problem is equivalent to $\widetilde{M}(\rho,\sigma)=[\lambda_{1}\mu_{2}+\lambda_{2}\mu_{1},\lambda_{1}\mu_{1}+\lambda_{2}\mu_{2}]$ for $\lambda_{1},\lambda_{2}$ and
$\mu_{1},\mu_{2}$ are  eigenvalues of $\rho,\sigma\in{\textmd{Pd}(\cX)}$ , respectively. For any $k\in[0,1]$ , there exists unitary $U$, such that
\begin{equation*}
U\rho{U^\ast}=\begin{bmatrix}
k\lambda_{1}+(1-k)\lambda_{2} & \sqrt{k}\sqrt{1-k}(\lambda_{1}-\lambda_{2}) \\ \sqrt{k}\sqrt{1-k}(\lambda_{1}-\lambda_{2}
& k\lambda_{2}+(1-k)\lambda_{1}
\end{bmatrix},
\end{equation*}
where
\begin{equation*}
U=\begin{bmatrix}
\sqrt{k} & \sqrt{1-k} \\  \sqrt{1-k} & -\sqrt{k}
\end{bmatrix},  k\in[0,1].
\end{equation*}
Then
\begin{eqnarray*}
\textmd{Tr}(U\rho{U^{\ast}}\sigma) & = & [(1-k)\lambda_{2}+k\lambda_{1}]\mu_{1}+[k\lambda_{2}+(1-k)\lambda_{1}]\mu_{2} \\
& = & (1-k)(\lambda_{1}\mu_{2}+\lambda_{2}\mu_{1})+k(\lambda_{1}\mu_{1}+\lambda_{2}\mu_{2}). \\
\end{eqnarray*}

It is equivalent to $\textmd{Tr}(U\rho{U^{\ast}}\sigma)$ is convex combination of  $\lambda_{1}\mu_{2}+\lambda_{2}\mu_{1}$  and  $\lambda_{1}\mu_{1}+\lambda_{2}\mu_{2}$. Therefore, when $d=2$, $\widetilde{M}(\rho,\sigma)=[\lambda_{1}\mu_{2}+\lambda_{2}\mu_{1},\lambda_{1}\mu_{1}+\lambda_{2}\mu_{2}]$.

Next we suppose that the conclusion holds for$d=n-1$, i.e. if $d=n-1$,
\begin{equation}
\begin{split}
\widetilde{M}(\rho,\sigma)=[\displaystyle\sum_{i=1}^{n-1} \lambda_{i}^\downarrow(\rho)\log\lambda_{i}^\uparrow(\sigma),\displaystyle\sum_{i=1}^{n-1} \lambda_{i}^\downarrow(\rho)\log\lambda_{i}^\downarrow(\sigma)]=[\displaystyle\sum_{i=1}^{n-1} \lambda_{i}\mu_{n-i},\displaystyle\sum_{i=1}^{n-1} \lambda_{i}\mu_{i}]
\end{split}
\end{equation}
where $\lambda^\downarrow(\rho)=(\lambda_{1},\ldots,\lambda_{n-1}) \quad and  \quad \lambda^\downarrow(\sigma)=(\mu_{1},\ldots,\mu_{n-1})$.

Now we should to prove the conclusion holds for $d=n$.

From (2.1), we know that
\begin{eqnarray*}
\displaystyle\sum_{i=1}^{n}\lambda_{i}^\downarrow(\rho)\log\lambda_{i}^\downarrow(\sigma)
&   =  & \lambda_{1}\mu_{1}+\lambda_{2}\mu_{2}+\lambda_{3}\mu_{3}+\ldots+\lambda_{n}\mu_{n} \\
& \geq & \lambda_{2}\mu_{1}+\lambda_{1}\mu_{2}+\lambda_{3}\mu_{3}+\ldots+\lambda_{n}\mu_{n} \\
& \geq & \lambda_{2}\mu_{1}+\lambda_{3}\mu_{2}+\lambda_{1}\mu_{3}+\ldots+\lambda_{n}\mu_{n} \\
& \geq & \ldots  \\
& \geq & \lambda_{2}\mu_{1}+\lambda_{3}\mu_{2}+\lambda_{4}\mu_{3}+\ldots+\lambda_{1}\mu_{n-1}+\lambda_{n}\mu_{n} \\
& \geq & \lambda_{2}\mu_{1}+\lambda_{3}\mu_{2}+\lambda_{4}\mu_{3}+\ldots+\lambda_{n}\mu_{n-1}+\lambda_{1}\mu_{n} \\
\end{eqnarray*}
For the first inequality,
\begin{eqnarray*}
\lambda_{1}\mu_{1}+\lambda_{2}\mu_{2}+\lambda_{3}\mu_{3}+\ldots+\lambda_{n}\mu_{n} \\
\geq  \lambda_{2}\mu_{1}+\lambda_{1}\mu_{2}+\lambda_{3}\mu_{3}+\ldots+\lambda_{n}\mu_{n}
\end{eqnarray*}
there exists an unitary $U$, for any $k\in[0,1]$, such that
\begin{equation*}
U\rho{U^\ast}=\begin{bmatrix}
k\lambda_{1}+(1-k)\lambda_{2} & \sqrt{k}\sqrt{1-k}(\lambda_{1}-\lambda_{2}) & 0 &  0 & \ldots & 0 & \\
\sqrt{k}\sqrt{1-k}(\lambda_{1}-\lambda_{2}& k\lambda_{2}+(1-k)\lambda_{1} & 0 &  0 & \ldots & 0 & \\
0 & 0 & \lambda_{3} & 0 & \ldots & 0 & \\
0 & 0 & 0 & \lambda_{4} & \ldots & 0 & \\
\ldots & \ldots & \ldots & \ldots & \ldots & \ldots & \\
0 & 0 & 0 & 0 & \ldots & \lambda_{n} & \\
\end{bmatrix},
\end{equation*}
where
\begin{equation*}
U=\begin{bmatrix}
\sqrt{k} & \sqrt{1-k} & 0  \\  \sqrt{1-k} & -\sqrt{k} & 0 \\
0& 0 & E \\
\end{bmatrix},
\end{equation*}
and $E$ is Unit matrix.

This is equivalent to $d=2$, thus $\textmd{Tr}(U\rho{U^{\ast}}\sigma)$ be full of the interval $[\lambda_{1}\mu_{2}+\lambda_{2}\mu_{1}+\displaystyle\sum_{i=3}^n \lambda_{i}\mu_{i},\lambda_{1}\mu_{1}+\lambda_{2}\mu_{2}+\displaystyle\sum_{i=3}^n \lambda_{i}\mu_{i}]$.
Therefore $\textmd{Tr}(U\rho{U^{\ast}}\sigma)$ be full of the interval
$$
[\displaystyle\sum_{i=1}^{k-1} \lambda_{i+1}\mu_{i}+\lambda_{1}\mu_{k}+\lambda_{k+1}\mu_{k+1}+\displaystyle\sum_{i=k+2}^n \lambda_{i}\mu_{i},\displaystyle\sum_{i=1}^{k-1} \lambda_{i+1}\mu_{i}+\lambda_{k+1}\mu_{k}+\lambda_{1}\mu_{k+1}+\displaystyle\sum_{i=k+2}^n \lambda_{i}\mu_{i}]
$$
where $k\in\{1,\ldots,n\}$.

Thus $\textmd{Tr}(U\rho{U^{\ast}}\sigma)$ be full of the interval
\begin{equation}
\begin{split}
[\displaystyle\sum_{i=1}^{n-1} \lambda_{i+1}\mu_{n-i}+\lambda_{1}\mu_{n},\displaystyle\sum_{i=1}^{n-1} \lambda_{i+1}\mu_{i}+\lambda_{1}\mu_{n}]
\end{split}
\end{equation}.

 From (3.1) we know that if $d=n-1$ , for $\lambda^\downarrow(\rho)=(\lambda_{2},\ldots,\lambda_{n}) \quad and  \quad \lambda^\downarrow(\sigma)=(\mu_{1},\ldots,\mu_{n-1})$,we have  $\widetilde{M}(\rho,\sigma)=[\displaystyle\sum_{i=1}^{n-1} \lambda_{i+1}\mu_{n-i},\displaystyle\sum_{i=1}^{n-1} \lambda_{i+1}\mu_{i}]$.Then when $d=n$ , $\textmd{Tr}(U\rho{U^{\ast}}\sigma)$ be full of the interval
 $[\displaystyle\sum_{i=1}^{n-1} \lambda_{i+1}\mu_{n-i}+\lambda_{1}\mu_{n},\displaystyle\sum_{i=1}^{n-1} \lambda_{i+1}\mu_{i}+\lambda_{1}\mu_{n}]$.

According to (3.1) and (3.2) , $\textmd{Tr}(U\rho{U^{\ast}}\sigma)$ be full of the interval
$$
[\displaystyle\sum_{i=1}^{n-1} \lambda_{i+1}\mu_{n-i}+\lambda_{1}\mu_{n},\displaystyle\sum_{i=1}^{n-1} \lambda_{i+1}\mu_{i}+\lambda_{1}\mu_{n}]\cup
[\displaystyle\sum_{i=1}^{n-1} \lambda_{i+1}\mu_{n-i}+\lambda_{1}\mu_{n},\displaystyle\sum_{i=1}^{n-1} \lambda_{i+1}\mu_{i}+\lambda_{1}\mu_{n}]
$$
Therefore $\textmd{Tr}(U\rho{U^{\ast}}\sigma)$ be full of the interval
$$
[\displaystyle\sum_{i=1}^{n-1} \lambda_{i+1}\mu_{n-i}+\lambda_{1}\mu_{n},\displaystyle\sum_{i=1}^{n-1} \lambda_{i+1}\mu_{i}+\lambda_{1}\mu_{n}]=
[\displaystyle\sum_{i=1}^{n} \lambda_{i}\mu_{n-i+1},\displaystyle\sum_{i=1}^{n}\lambda_{i}\mu_{i}]
$$

According to (2.2) and (2.3) ,
$$
\displaystyle\max_{U\in{U(\cX_{d})}} \textmd{Tr}(U\rho{U^{\ast}}\sigma)=\displaystyle\sum_{j=1}^{n} \lambda_{j}^\downarrow(\rho)\lambda_{j}^\downarrow(\sigma)=\displaystyle\sum_{i=1}^{n}\lambda_{i}\mu_{i}
$$
and
$$
\displaystyle\min_{U\in{U(\cX_{d})}} \textmd{Tr}(U\rho{U^{\ast}}\sigma)=\displaystyle\sum_{j=1}^{n} \lambda_{j}^\downarrow(\rho)\lambda_{j}^\uparrow(\sigma)=\displaystyle\sum_{i=1}^{n} \lambda_{i}\mu_{n-i+1}.
$$
Therefore $\widetilde{M}(\rho,\sigma)=[\displaystyle\sum_{i=1}^{n} \lambda_{i}\mu_{n-i+1},\displaystyle\sum_{i=1}^{n}\lambda_{i}\mu_{i}]$.
The proof is completed.
\end{proof}

Therefore we have
\begin{eqnarray*}
M(\rho,\sigma) & = & [\displaystyle\min_{U\in{U(\cX_{d})}} S(U\rho{U^{\ast}}\parallel\sigma),\displaystyle\max_{U\in{U(\cX_{d})}} S(U\rho{U^{\ast}}\parallel\sigma)] \\
& = & [-S(\rho)-\langle\lambda^\downarrow(\rho),\log\lambda^\downarrow(\rho)\rangle,
-S(\rho)-\langle\lambda^\downarrow(\rho),\log\lambda^\uparrow(\rho)\rangle].\\
\end{eqnarray*}
Then $S(U\rho{U^{\ast}}\parallel\sigma)$ be full of the interval $[\displaystyle\min_{U\in{U(\cX_{d})}} S(U\rho{U^{\ast}}\parallel\sigma),\displaystyle\max_{U\in{U(\cX_{d})}} S(U\rho{U^{\ast}}\parallel\sigma)].$

\section{A generalization}

In fact, the above two optimal problems can be easily generalized as follows:

Given two operators $\rho,\sigma\in{\textmd{Pd}(\cX)}$ are positive definite operators on $\cX$. If $\omega\prec\rho$, then
$$
max/min\{S(\omega\parallel\sigma):\omega\prec\rho\}=max/min\{S(\Phi(\rho)\parallel\sigma):\Phi~\textmd{bi-stochastic}\}
$$
Thus we have
$$
\displaystyle\max_{\Phi~ \textmd{bi-stochastic}}{S(\Phi(\rho)\parallel\sigma)}=-S(\rho)-\langle\lambda^\downarrow(\rho),\log\lambda^\uparrow(\rho)\rangle
$$

and
$$
\displaystyle\min_{\Phi~\textmd{bi-stochastic}}{S(\Phi(\rho)\parallel\sigma)}=
-S(\rho)-\langle\lambda^\downarrow(\rho),\log\lambda^\downarrow(\rho)\rangle.
$$
If we denote
$$
M(\rho,\sigma)\defeq \{S(\Phi(\rho)\parallel\sigma):\Phi~\textmd{bi-stochastic}\}
$$
Similarly,we have
$$
M(\rho,\sigma)=[-S(\rho)-\langle\lambda^\downarrow(\rho),\log\lambda^\downarrow(\rho)\rangle,
-S(\rho)-\langle\lambda^\downarrow(\rho),\log\lambda^\uparrow(\rho)\rangle].
$$


\subsection*{Acknowledgement}

The author would like to thank the anonymous referee(s) for valuable
comments and to thank  *** for pointing out some misprints
appearing in this paper.The author is also grateful for funding from
Hangzhou Dianzi University .


\end{document}